\documentclass[11pt]{article}
\usepackage[utf8]{inputenc}
\usepackage{amsmath,amssymb,graphicx,tikz,amsthm,hyperref,algorithm,algpseudocode,a4wide}

\newtheorem{theorem}{Theorem}
\newtheorem{lemma}{Lemma}
\newtheorem{conjecture}{Conjecture}

\begin{document}


\title{Scheduling with a processing time oracle\thanks{This work was partially supported by the French research agency (ANR-18-CE25-0008 and ANR-19-CE48-0016), by the EPSRC grant  EP/S033483/1 as well as by ANID, Proyecto FONDECYT 1211640 and Universidad de Santiago de Chile, Beca Olga Ulianova.}}

\author{
Fanny Dufossé\thanks{Univ. Grenoble-Alpes, Inria, CNRS, LIG, France}
\and 
Christoph Dürr\thanks{Sorbonne University, CNRS, LIP6, Paris, France}
\and 
Noël Nadal\thanks{Ecole Normale Supérieure Paris-Saclay, France}
\and
Denis Trystram\footnotemark[2]
\and
Óscar C. Vásquez\thanks{Industrial Engineering Department, Universidad de Santiago de Chile, Chile}
}

\maketitle

\begin{abstract}   In this paper we study a single machine scheduling problem
with the objective of minimizing the sum of completion times. Each of the
given jobs is either short or long.  However the processing times are
initially hidden to the algorithm, but can be tested. This is done by
executing a processing time oracle, which reveals the processing time of a
given job.  Each test occupies a  time unit in the schedule, therefore the
algorithm must decide for which jobs it will call the processing time oracle. 
The objective value of the resulting schedule is compared with the objective
value of an optimal schedule, which is computed using full information.  The
resulting competitive ratio measures the price of hidden processing times, and
the goal is to design an algorithm with minimal competitive ratio.

Two models are studied in this paper. In the \emph{non-adaptive} model, the
algorithm needs to decide beforehand which jobs to test, and which jobs to
execute untested.  However in the \emph{adaptive} model, the algorithm can
make these decisions adaptively depending on the outcomes of the job tests.  
In both models we provide optimal polynomial time algorithms following a
\emph{two-phase strategy}, which consist of a first phase where jobs are
tested, and a second phase where jobs are executed obliviously.  Experiments
give strong evidence that optimal algorithms have this structure. Proving this
property is left as an open problem.  \end{abstract}


\section{Introduction}

A typical combinatorial optimization problem consists of a clear defined
input, for which the algorithm has to compute a solution minimizing some cost.
This is the beautiful simple world of theory. In contrast, everyone who
participated in some industrial project can testify that obtaining the input
is one of the hardest aspects of problem solving.  Sometimes the client does
not have the precise data for the problem at hand, and provides only some
imprecise estimations.  However, imprecise input can only lead to imprecise
output, and the resulting solution might not be optimal.  Different approaches
have been proposed to deal with this situation, such as robust optimization or
stochastic optimization,
see \cite{garcia_robust_2018,gorissen_practical_2015,marti_stochastic_2015} for
surveys on those areas.  

In this paper we follow the paradigm of \emph{optimizing under explorable
uncertainty}, which is an alternative approach which has been introduced in
1991 \cite{kahan_model_1991}, and started to be applied to scheduling problems
in 2016 \cite{levi_scheduling_2016}.   In this approach, a problem instance
consists of a set of numerical parameters, the algorithm obtains as input only
an uncertainty interval for each one. The algorithm knows for each parameter
that it belongs to the given interval and has the possibility to make a query
in order to obtain the precise value.  Clearly a compromise has to be found
between the number of queries an algorithm makes and the quality of the
solution it produces.   This setting differs from a probabilistic one, studied
in \cite{levi_scheduling_2016,levi_scheduling_2018}, where jobs have weights
and processing times drawn from known distributions, and the algorithm can
query these parameters.  A seemingly similar problem has been studied in
\cite{durr_scheduling_2018,durr_adversarial_2020,albers_explorable_2020},
where the term \emph{testing} has a different meaning than in this paper.  



Different measures to evaluate the performance of an algorithm under
explorable uncertainty have been investigated.  Early work studied the number
of queries required in order to be able to produce an optimal solution, no
matter what the values of the non queried parameters happen to be.  In this
sense the queries are used to form a proof of optimality, and the underlying
techniques are close to the ones used in query complexity.

A broad range of problems has been studied, namely the finding the median 
\cite{feder_computing_2000}, or more generally the $k$-th smallest value among
the given parameters \cite{gupta_adaptivity_2017}, sorting
\cite{chaplick_query_2020}, finding a shortest path when the uncertain
parameters are the edge lengths \cite{feder_computing_2007}, determining the
convex hull of given uncertain points \cite{bruce_efficient_2005}, or
computing a minimum spanning tree
\cite{erlebach_computing_2008,focke_minimum_2017,megow_randomization_2017}. 
The techniques which have been developed in these papers have been generalized
and described in
\cite{erlebach_query-competitive_2016,erlebach_computing_2018}.

A related stochastic model for scheduling with testing has been introduced by
Levi, Magnanti and
Shaposhnik~\cite{levi_scheduling_2018,shaposhnik_exploration_2016}.  They
consider the problem of minimizing the weighted sum of completion times on one
machine for jobs whose processing times and weights are random variables with
a joint distribution, and are independent and identically distributed across
jobs. In their model, testing a job provides information to the scheduler (by
revealing the exact weight and processing time for a job, whereas initially
only the distribution is known). They present structural results about optimal
policies and efficient optimal or near-optimal solutions based on dynamic
programming.

\section{Our Contribution} 
\label{sec:contribution}

\begin{figure}[!ht]
\centerline{\includegraphics[width=\textwidth]{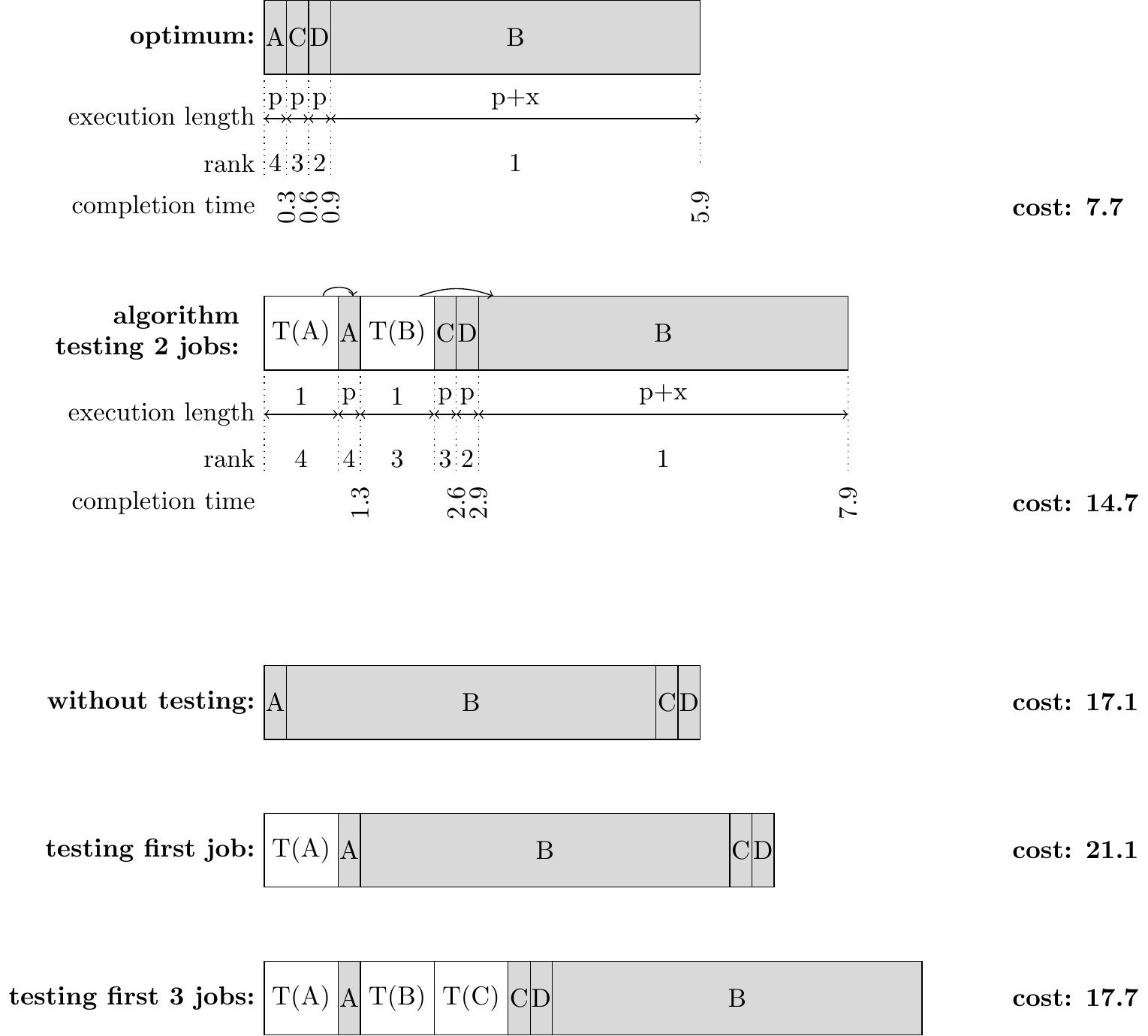}}
\caption{Some schedules with four jobs $A,B,C,D$ and parameters $p=0.3, x=4.7$.  
White boxes represent job tests, while gray boxes represent job executions. The cost is the sum of completion times or equivalently the sum of  the execution lengths of the boxes, multiplied with their respective ranks.
}
\label{fig:example_ABCD} 
\end{figure}

Before defining formally the problem, we start with an illustrative example.
Suppose that you need to schedule on a single machine four jobs called
$A,B,C,D$. Job $B$ has processing time 5, while jobs $A,C,D$ have processing
times $0.3$, see Figure~\ref{fig:example_ABCD}.  The cost of a schedule (also
called \emph{objective value}) is the total completion time of the jobs. Hence
one possible optimal execution order is $A,C,D,B$, which has objective value
$7.7$. However the processing times are initially unknown to you, the jobs are
indistinguishable, and you only know that the processing times are either
$0.3$ or $5$.  So you could schedule them in an arbitrary order, for example
the execution order $A,B,C,D$ would yield an objective value of $17.1$, which
is roughly $2.22$ larger than the optimum. Luckily you have access to a
processing time oracle, which allows you to query the processing time of a
particular job, and this operation takes $1$ time unit.  Such an oracle can be
thought as a machine learning black box predictor, which was trained on large
number of jobs.   Throughout the paper we use indistinguishable the word
\emph{test} or \emph{query} for this operation.  You could query this oracle
for all 4 jobs, providing full information, which allows you to schedule the
jobs in the optimal order.  However you
don't have to wait until you know the processing time of all jobs before you
start executing them.  For example if a query reveals a short job, then it could be
scheduled immediately, and if a query reveals a long job, then its execution
could be postponed towards the end of the schedule.   This means that there is
no benefit to query the last job, as the position of its execution would be
the same, regardless of the outcome. In summary, for this example you have the
possibility to query between 0 and 3 jobs. The resulting schedules are
depicted in Figure~\ref{fig:example_ABCD}. Testing the first 2 jobs would lead
to the smallest objective value of $14.7$, which is only roughly $1.91$ larger
than the optimum.



The general problem studied in this paper consists in scheduling $n$
independent jobs on a single machine, with jobs duration limited to two
possible values, $p$ and $p+x$, for some given parameters $p,x > 0$. Each job can
be tested, requiring one time unit on the schedule, revealing the actual
processing time. The objective value of a schedule is the sum of the
completion times of the jobs, and is also called the \emph{cost} of the
schedule.  This value is normalized by dividing it with the minimum cost over
all schedules, which is called the \emph{optimum}. Computing the optimum
requires full knowledge of the job's processing times. This ratio is called
the \emph{competitive ratio}, and the goal is to design an algorithm with
smallest possible competitive ratio. The ratio measures the price of
not knowing initially all processing times.

We emphasize on the importance of normalizing the objective value. The worst
case instance for the sum of completion times, consists of only long jobs, and
any test would strictly increase the objective value.  In other words, if the
goal were to minimize the objective value in the worst case (and not the
competitive ratio), then the best strategy would execute all jobs
untested.

Competitive analysis has been introduced during the 1980's to study
computational models where the input is a request sequence given sequentially
to the algorithm \cite{sleator_amortized_1985}.  Our problem does not fall
exactly into this setting, nevertheless we borrow its terminology.

Competitive ratio is usually seen as the value of a game played between an
algorithm deciding which jobs to test and an adversary deciding the job
lengths, hence generating the instance.  This is a zero sum game, and therefore it
admits an optimal algorithm and a schedule produced at equilibrium.

We study the competitive ratio in two algorithmic models. In the \emph{non
adaptive model}, the algorithm has to decide at the beginning of the
schedule how many jobs it wants to test,  while in the \emph{adaptive model}, it
can make this decision adaptively, depending on the outcome of previous tests.

At the beginning of the schedule, jobs are indistinguishable to the algorithm,
therefore we suppose that the algorithm processes the jobs in a fixed
arbitrary order.  As we show in Lemma~\ref{lem:postponing},  an optimal
algorithm  executes a tested short job right after its test, and delays the
execution of a tested long job towards the end of the schedule.  

With this observation in mind, the behavior of an algorithm can be fully
described by a binary decision for each job, namely to execute it untested, or
to test it and to execute it accordingly to the outcome as described above.  
From experiments, it seems that optimal algorithms follow a \emph{two-phase
strategy}, namely to first test some number of jobs, and then to execute the
remaining jobs untested (see Section~\ref{sec:experiments}). This does not specify where the tested jobs are executed, even though the optimal algorithm executes them in a specific manner (see previous paragraph and Lemma~\ref{lem:postponing}). 
In the example of
Figure~\ref{fig:example_ABCD} the first two jobs are tested and the the last
two jobs are executed untested.  Intuitively, this makes sense, as the benefit
of a test decreases with the progression of the schedule.   The purpose of a
test is to identify long jobs so they can be scheduled as late as possible. 
This allows to decrease the objective value, but the benefit depends on the
number of jobs whose processing time is not yet known.  Being unable to
provide a formal proof, we leave the proof of this conjectured dominant
behavior as an open problem, and focus in this paper only on two-phase
strategies.

\begin{conjecture} \label{conjecture}
	For all values of $p,x,n$, and for both the adaptive and non-adaptive models, 
	there is an optimal algorithm following a two-phase strategy.
\end{conjecture}

Here by \emph{optimal algorithm} we mean an algorithm which achieves the
smallest competitive ratio for the worst case instance, i.e.\ when the
adversary plays optimally.  In case the adversary does not play optimally, 
in the adaptive model the algorithm can achieve an even better competitive
ratio, but might need to diverge from a two-phase strategy for this.

For the non-adaptive model, we were able to provide an algorithm running in
time $O(n^2)$ which determines the optimal two-phase strategy.   Note that
this procedure is not polynomial in the input size, which consists only of the
values $p,x,n$, but it is polynomial in the number of jobs, and therefore also
in the size of the produced schedule. In addition we provide a closed form
expression of the limit of the competitive ratio, when $n$ tends to infinity.

For the adaptive model, we were also able to provide an algorithm running in
time $O(n^3)$.  More precisely, if there are still $r$ jobs to be handled, the
algorithm computes  in time $O(r^3)$ a strategy for the remainder (how many
jobs it wants to test).  The algorithm will stick to this strategy if the
adversary behaves optimally.   However once the adversary diverges from its
optimal strategy, the algorithm needs to recompute the strategy for the
remaining jobs.  This means that if the algorithm plays against the optimal
adversary it would spend time $O(n^3)$ to decide whether to test the first
job, and $O(1)$ for every subsequent job, otherwise it has time complexity
$O(n^3)$ on every job decision.

\section{Problem settings}

Formally, we consider the problem of scheduling $n$ independent jobs on a
single processor with the objective of minimizing the sum of completion times.
Jobs are numbered from 1 to $n$.  Every job can be either short (processing
time $p$) or long (processing time $p + x$) for some known parameters $p, x
> 0$. The algorithm receives $n$ jobs without the information of their
processing times and will handle them in order of their indices.  For every
job, the algorithm has the choice to test it or to execute it untested. A job
test consists in the execution of a processing time oracle whose duration is
one time unit and reveals to the algorithm the processing time of the job.  
In principle a tested job could be scheduled at any moment in the schedule,
but as we show in Lemma~\ref{lem:postponing}, it is dominant to schedule short
jobs immediately after their test and to postpone long tested jobs towards the
end of the schedule.

An example is pictured in Figure~\ref{fig:example_ABCD}. Here, the scheduler
decides to test the first two jobs. The first test indicates a short job that
is executed immediately, and the second test indicates a long job that is
postponed to the end of the schedule. The last two jobs are
executed untested.

Borrowing the terminology of online algorithms we consider the problem as a
game played between an algorithm and an adversary.  See
\cite{albers_online_2003,komm_introduction_2016} for recent introductions into
online algorithms.  But we emphasize that this is not an online problem, in
the sense that the algorithm receives from the beginning on all $n$ jobs.   It
chooses a strategy $u\in\{T,E\}^n$, while the adversary chooses a strategy
$v\in\{p,x\}^n$. Here we abuse notation and used symbols $T,E,p,x$, even
though $p,x$ are also used for the numerical parameters of the game. For every
job $j\in\{1,\ldots,n\}$, if $u_j=T$, then the algorithm tests the job,
otherwise the algorithm executes the job untested. If $v_j=p$ then job $j$ is
short, otherwise it is long.  The resulting competitive ratio is a function of
the strategies $u$ and $v$, and is formally defined by the pseudocode given in
Algorithm~\ref{fig:CR}. A two-phase strategy $u$ has the shape $T^*E^*$.

\begin{algorithm}
\begin{algorithmic}[0]
	\State{ALG = pn(n+1)/2}  \Comment{base cost}
	\State{$r = n$} \Comment{rank of next job to be processed}
	\For{$j=1,\ldots,n$}	\Comment{Process jobs in given order}
		\If{$u_j=E$ and $v_j=p$} \Comment{execute short untested job}
		\State{$r = r - 1$}
		\EndIf
		\If{$u_j=E$ and $v_j=x$}\Comment{execute long untested job}
		\State{ALG = ALG + $r x$}
		\State{$r = r - 1$}
		\EndIf
		\If{$u_j=T$ and $v_j=p$}\Comment{test short job and execute}
		\State{ALG = ALG + $r$}
		\State{$r = r - 1$}
		\EndIf
		\If{$u_j=T$ and $v_j=x$}\Comment{test long job and postpone}
		\State{ALG = ALG + $r$}
		\EndIf
	\EndFor
	\State{ALG = ALG + $x r ( r + 1)  / 2 $} 	\Comment{execute postponed long tested jobs}
	\State{$\ell = |\{j:v_j=x\}|$}	\Comment{number of long jobs}
	\State{OPT = $( p n(n+1) + x \ell(\ell+1))/2$}
	\State \Return{ ALG / OPT } 
\end{algorithmic}
\caption{Pseudocode defining the competitive ratio as a function of the
strategies $u\in\{T,E\}^n, v\in\{p,x\}^n$. ALG represents the cost of the
schedule produced by the algorithm and OPT the optimal cost.}

\label{fig:CR}
\end{algorithm}

The algorithm wants to minimize the competitive ratio, while the adversary
wants to maximize it.  In the non-adaptive model, the algorithm plays $u$, and
the adversary plays $v$ in response, while in the adaptive model, algorithm
and adversary play in alternation: First the algorithm plays $u_1$, then the
adversary plays $v_1$, then the algorithm plays $u_2$ and so on.

For convenience we describe a schedule with the string $u_1 v_1 \ldots u_n
v_n$, which encodes the decisions made by both the algorithm and the
adversary, and compactly describes the actual schedule. Such a string matches
the regular expression $((T|E)(p|x))^n$. For example the schedule resulting of testing 2 jobs in 
Figure~\ref{fig:example_ABCD} would be described by $TpTxEpEp$. When both the
algorithm and the adversary play optimally, then we call the resulting
schedule the \emph{equilibrium schedule}.  In this notation $Ep$ describes the
execution of an untested short job, $Ex$ the execution of a long untested job,
$Tp$ the test of a short job, followed immediately by its execution and $Tx$
the test of a long job, whose execution is delayed towards the end of the
schedule and does not explicitly appear in the notation.



\subsection*{Rank}

In this paper it will often be convenient to express the cost of a schedule
--- the total completion time --- using the notion of \emph{rank}.  A schedule
consists of a sequence of job tests and job executions. Each of these actions
has a rank which is defined as the number of jobs which are executed after
this part, including the action itself in case of a job execution.  This
permits us to express the total job completion time as the sum of all actions
in the schedule of the length of the action multiplied by its rank.  For
example in Figure~\ref{fig:example_ABCD} (algorithm testing two jobs), the
first test duration delays the completion time of all 4 jobs, it thus has rank
4. The second test delays all but the first job, therefore it has rank 3.  The
tested long job is postponed to the end of the schedule, and has rank 1.  The
two last jobs which are executed untested have rank respectively 3 and~2.

\subsection*{Two algorithmic models}

We study two models. In the \emph{adaptive model}, the algorithm can adapt
to the adversary after each step. This setting is studied in
Section~\ref{sec:adaptive}. In contrast, in the \emph{non-adaptive model} the
algorithm has to decide once for all on a sequence of testing and executing,
and stick to it, no matter what the job lengths happen to be. The results in
this setting are presented in Section~\ref{sec:nonadaptive}.


See Table~\ref{tab:example} for an illustration of the non-adaptive model. The
algorithm chooses a particular strategy (column) and the adversary chooses a
particular response (row). The resulting ratio is indicated in the selected
cell of this array. The equilibrium schedule is determined by the min-max
value of this array, namely $ExEp$ in this case, resulting in the ratio 11/7.

In the adaptive model, the interaction between the algorithm and the adversary
is illustrated by the game tree shown in Figure~\ref{fig:tree_example}. Every
node represents a particular moment of the interaction. Leaf nodes are labeled
with the resulting ratio. Inner nodes have two out-going arcs, representing
the possible actions of the algorithm or the adversary, which play in
alternation, and are labeled with the ratio resulting of a best choice. The
algorithm is the minimizer in this game, while the adversary is the maximizer.
The root is labeled with the competitive ratio of the game, which is 11/7 in our
example.  For this small example, the equilibrium schedule happen to be the
same in both the adaptive and the non-adaptive model, but for larger number of
jobs, the decision tree would be illegible.

\begin{table*}
$$
\arraycolsep=3pt\def\arraystretch{2}
\begin{array}{|c|c|c|c|c|}
\hline
&EE&ET&TE&TT\\
\hline
pp&\frac{3p}{3p}=1&\frac{3p+1}{3p}=\frac{4}{3}&\frac{3p+2}{3p}=\frac{5}{3}&\frac{3p+3}{3p}=2\\
\hline
px&\frac{3p+x}{3p+x}=1&\frac{3p+x+1}{3p+x}=\frac{8}{7}&\frac{3p+x+2}{3p+x}=\frac{9}{7}&\frac{3p+x+3}{3p+x}=\frac{10}{7}\\
\hline
xp&\frac{3p+2x}{3p+x}=\frac{11}{7}&\frac{3p+2x+1}{3p+x}=\frac{12}{7}&\frac{3p+x+2}{3p+x}=\frac{9}{7}&\frac{3p+x+4}{3p+x}=\frac{11}{7}\\
\hline
xx&\frac{3p+3x}{3p+3x}=1&\frac{3p+3x+1}{3p+3x}=\frac{16}{15}&\frac{3p+3x+2}{3p+3x}=\frac{17}{15}&\frac{3p+3x+4}{3p+3x}=\frac{19}{15}\\
\hline
\end{array}
$$
    \caption{The competitive ratio in the non-adaptive setting for 2 jobs, $p=1$ and $x=4$, as a function of the algorithm's strategy (columns) and the adversarial strategy (rows).}
    \label{tab:example}
\end{table*}
\begin{figure*}
    \centerline{\includegraphics[width=9cm]{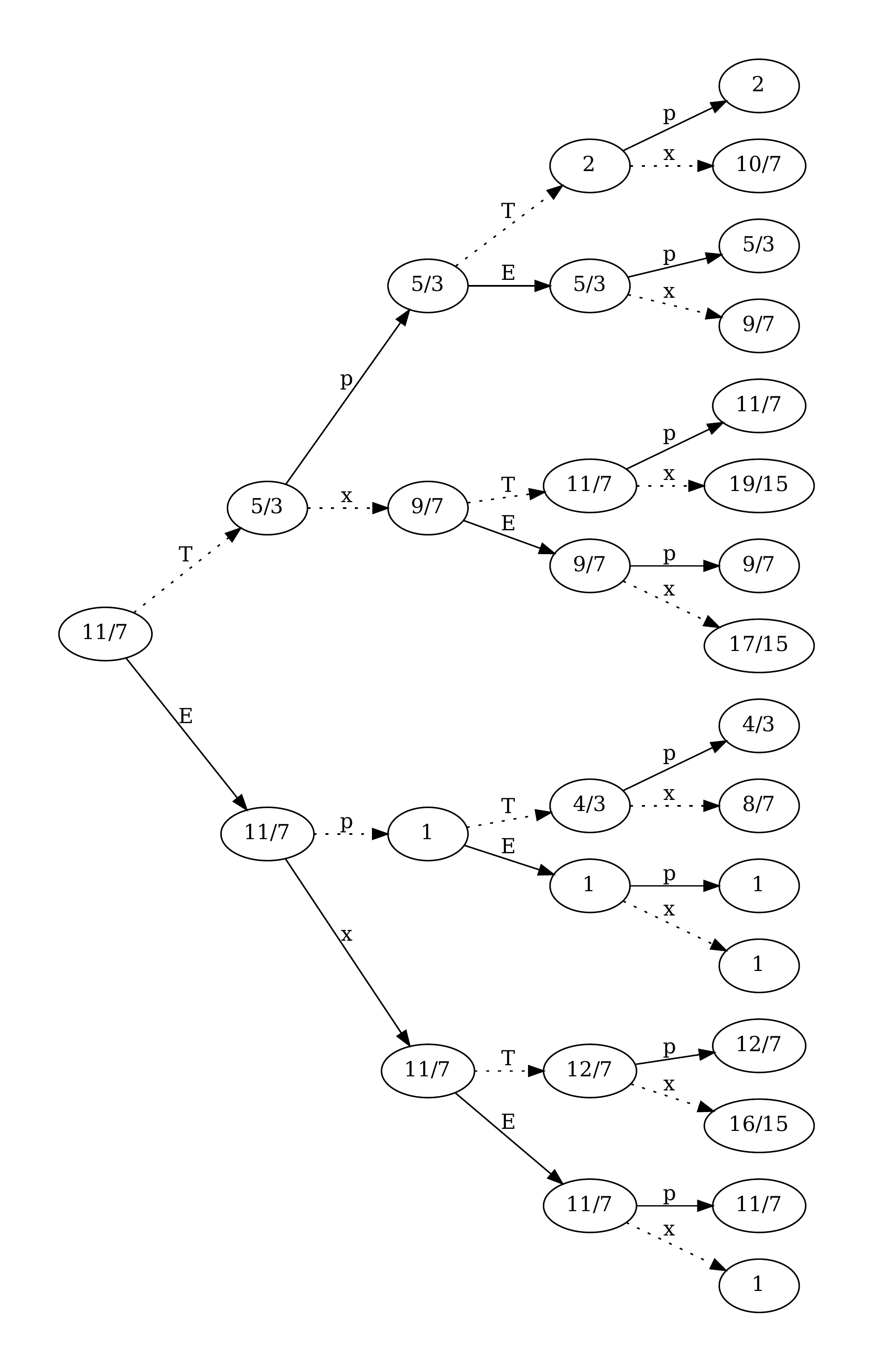}}
  \caption{Decision tree in the adaptive setting for 2 jobs, $p=1$ and $x=4$. Solid arcs indicate the optimal choice for the player in turn.}
  \label{fig:tree_example}
\end{figure*}

\section{Dominance properties}

The algorithm processes jobs in given order. At any moment it can decide to test the next job in this order, or to execute it untested, or to execute a previously tested job.  In this section we show some dominant behavior of optimal algorithms.  

We say that an algorithm is \emph{postponing} if it executes all tested long jobs at the end of the schedule and executes every tested short job immediately after its test.  We will show that any optimal algorithm has this property.  In addition we say that the adversary is \emph{restricted} if it is forced to make the last job long.

We consider the game tree resulting of an adaptive algorithm (not necessarily postponing) playing against an optimal adversary. For every node where the algorithm is in turn, 3 actions are possible, namely to test a job, to execute an untested job or to execute a tested long job. After the first two actions the adversary is in turn, while after the last action the algorithm is still in turn.  Figure~\ref{fig:tree_example} shows only the first two actions.

Every node has a ratio associated to it, which is the competitive ratio reached from this node, when both the algorithm and the adversary play optimally. Consider a node in the game tree, where the algorithm is in turn, and assume that the algorithm did not yet execute untested jobs.  The ratio of this node is completely described by some parameters $c,d,e,f$: $c$ is the number of tested short jobs, $d$ the number of tested long jobs which are pending (not yet executed), $f$ is the number of untested jobs, and $e$ is a value determining the algorithms cost in the following way. It is defined as the total rank over all tests plus $x$ times the total rank of all executions of long tested jobs.  If this node is a leaf in the game tree, the algorithm's cost is $e+xd(d+1)/2+pn(n+1)/2 $, where we added to $e$ the additional cost of the postponed long job executions, and a \emph{base cost} generated by all $n$ job executions.  We use the following notations for the ratio of a node with parameters $c,d,e,f$, under various restrictions on the players.

\begin{quote}
    \begin{tabular}{lll}
        notation & algorithm & adversary \\ \hline
        $\textrm{ratio}(c,d,e,f)$ & unrestricted & unrestricted \\
        $\textrm{ratio'}(c,d,e,f)$ & postponing & unrestricted \\
        $\textrm{ratio''}(c,d,e,f)$ & postponing & restricted \\
    \end{tabular}
\end{quote}

We can show the following useful inequality.

\begin{lemma}\label{lem:restricted}
	For any $c,d\geq 0$, $e>0$ and $f \geq 1$ we have 
	$\textrm{ratio'}(c,d,e,f) > \textrm{ratio'}(c,d+1,e,f-1)$.
\end{lemma}
\begin{proof}
First we observe $\textrm{ratio''}(c,d,e,f) \leq \textrm{ratio'}(c,d,e,f)$, which
holds simply because a restricted adversary has less choices than an
unrestricted adversary. More arguments are necessary to show that the inequality is strict.  
Hence consider the moment $t$ when the algorithm processes the last job (either by testing it or by executing it untested). No matter what the action of the algorithm is for this job, the job will be executed before all $d$ pending long jobs.  Then if the adversary makes this last job short instead of long, the cost of the algorithm decreases by exactly $dx$, while the cost of the adversary decreases by at least $xd$, in fact by $x(d+d')$, where $d'$ is the number of long jobs which have been executed by the algorithm before $t$. As a result, the competitive ratio is strictly increasing,  unless the algorithm and the adversary produce exactly the same schedules. This would mean that the algorithm tested no job, and executed no long job before time $t$, which contradicts the assumption $e>0$.

Next we observe the equality $\textrm{ratio''}(c,d,e,f)=\textrm{ratio'}(c,d+1,e,f-1)$. The point is that the last processed job $j$ executed in the equilibrium schedule corresponding to $\textrm{ratio''}(c,d,e,f)$ is long and followed immediately by the $d$ tested postponed long jobs. Hence conceptually one could consider this job as an additional tested long job, decreasing at the same time the number of untested jobs. Because the restricted adversary was only committed to job $j$, and not to the other $f-1$ jobs, the equilibrium schedule corresponding to $\textrm{ratio''}(c,d,e,f)$ is identical to the one corresponding to $\textrm{ratio'}(c,d+1,e,f-1)$. This concludes the proof.
\end{proof}

The previous lemma allows us to show that the optimal algorithm is postponing.

\begin{lemma}\label{lem:postponing}
	Without loss of generality
	the optimal algorithm is postponing, in other words ratio=ratio'.
	This holds for both the adaptive and the non-adaptive model and under Conjecture~\ref{conjecture}.	
\end{lemma}
\begin{proof}
We provide the proof only for the adaptive model, as it is the hardest.  First we observe that a tested short job $j$ can be safely executed by the algorithm right after its test.  This choice reduces the rank of the following actions happening between this moment and the eventual execution of job $j$, and has no influence on the adversary's strategies.  Hence the optimal algorithm executes tested short jobs right after their test.

For the tested long jobs, formally we show that for every node in the game tree with parameters $c,d,e,f$, the optimal algorithm will not execute a tested long job as the next action.  The proof is by induction on the pair $(f,d)$ in lexicographical order.   The base case $f=0$, holds trivially. Indeed the node is in fact a leaf in the game tree, and the algorithm has no option but to execute all $d$ tested long jobs.

For the induction step consider a node with parameters $c,d,e,f$ such that $f\geq 1$ and assume the induction hypothesis for all nodes with parameter $f-1$, or with parameter $f$ and $d-1$ in case $d\geq 1$.

We start with the easy case $d=0$. When there are no tested long jobs available, the algorithm has no choice but to test or to execute the next one among the $f$ untested jobs. This action results in a node with parameter $f-1$, and by induction hypothesis the algorithm is postponing.

It remains to show the induction step for the case $d\geq 1$.
The following table shows the ratios resulting by an action from the algorithm followed possibly by an action from the adversary.

\begin{center}
\begin{tabular}{lll}
short & action & ratio \\ 
\hline
$Tp$ & test a short job & $\textrm{ratio}(c+1,d,e+f+d,f-1)$
\\
$Tx$ & test a long job & $\textrm{ratio}(c,d+1,e+f+d,f-1)$\\
$Ep$ & execute an untested short job & $C=\textrm{ratio}(c+1,d,e,f-1)$ \\
$Ex$ & execute an untested long job & $B=\textrm{ratio}(c,d,e+x(f+d),f-1)$ \\
& execute a tested long job & $A=\textrm{ratio}(c,d-1,e+x(f+d),f)$
\end{tabular}
\end{center}

Ratios $A,B$ can be compared using the fact that by induction hypothesis for
nodes with parameters $f-1,d$ or $f,d-1$ the optimal algorithm is postponing,
which allows to use Lemma~\ref{lem:restricted}, with $d-1$ instead of $d$ to
match the expressions. It shows that ratio $A$ is larger than ratio
$B$.

But this is not enough, we need to show $A \geq \max\{B,C\}$.  This inequality would
certify that the optimal algorithm does not choose to execute a tested long
job now, which is the induction step to show.

When $B\leq C$ there is nothing to show. Hence  we assume $C>B$, in other
words the optimal adversary answers the E action (execute an untested job) by
a short job.  We claim that from now on all executions of untested jobs will
be short. Here we use the two-phase assumption.  This assumption states that
once an untested job is executed, all remaining untested jobs will be executed
untested as well.  As a result the subsequent schedule consists of $f$ job
executions followed by $d$ executions of the pending tested long jobs.  The
two-phase assumption really means that from now on the algorithm commits to
this execution pattern, hence the adversary can decide on a number $0\leq
b\leq f$ of jobs to be long. Which of the $f$ executed jobs will be long is of
no influence on the cost of the adversary, but placing them first maximizes
the cost of the algorithm.  This proves the claim.

Now if the algorithm knows that its remaining $f$ executions of untested jobs
will be short, then we claim that it prefers to execute an untested job,
rather than a tested long job.  The formal reason is that if the algorithm
decides for the later, then it will have a ratio strictly larger than $C$
already if the adversary still decides to make all untested jobs short, no
matter if they are tested or executed untested.

This concludes the induction step, and therefore the proof.
\end{proof}

We complete this section by describing what is happening at the last node in the game tree.

\begin{lemma}\label{lem:suffix_Ep}
	Consider a node in the game tree where the algorithm is in turn and a single untested job $j$ is left, that is a node with parameters $c,d,e$ and $f=1$.  Then the optimal algorithm executes job $j$ untested, following by all $d$ tested long jobs.  Moreover the adversary will make job $j$ short.  This holds for both the adaptive and the non-adaptive model and under Conjecture~\ref{conjecture}.	
\end{lemma}
\begin{proof}
By the previous lemma, the optimal algorithm is postponing, meaning that it
executes all $d$ tested long jobs after job $j$. Hence testing $j$ would not
improve the schedule. Executing it untested decreases
strictly the objective value, while leaving the adversary with the same set of
possible actions. 

If the adversary makes $j$ short instead of long, then he decreases the cost
of the algorithm by $x(d+1)$ and decreases his cost by at least $x(d+1)$. 
Hence making $j$ short increases strictly the competitive ratio and is
therefore the choice of the optimal adversary. 
\end{proof}

\section{Non adaptive algorithms}
\label{sec:nonadaptive}

In this section we analyze algorithms in the non-adaptive model following a
two-phase strategy. The set of algorithmic strategies can be modeled as
$u\in\{ T^a E^{n-a} : 0\leq a\leq n\}$, and the set of adversarial strategies
modeled as $v\in\{x,p\}^n$. The resulting schedule is described by the string
$u_1v_1u_2v_2\ldots u_n v_n$.

In addition we show some structure of the adversarial best response to any algorithmic strategy.

\begin{lemma}
	The adversarial best response $v$ to any algorithmic strategy $u$ of the form\footnote{We use regular expressions to describe the form of a strategy and the form of a schedule.}
	$T^* E^*$ results in a schedule of the form $(Tx)^*(Tp)^*(Ex)^*(Ep)^*$.
\end{lemma}
\begin{proof}
	We show the claim by an exchange argument. Let $v$ an adversarial strategy
	which is not of the claimed form. Then there is a position $0\leq i < n$ such
	that $u_iv_iu_{i+1}v_{i+1}\in\{TpTx,EpEx\}$.  We consider the effect of
	swapping $v_i$ with $v_{i+1}$ in the adversarial strategy. Since the number
	of long jobs is preserved, in order to analyze the change in the ratio, in
	fact we need to analyze the cost of the resulting schedule.  We observe that
	the cost is increased by the exchange, namely by $x$ in the case $EpEx$ and by
	1 in the case $TpTx$.
\end{proof}

The implication of this lemma is that there are only
$O(n^3)$ possible outcomes of the game. One can compute the ratios of these schedules in
amortized constant time by considering them in an appropriate order and
updating the costs of the algorithm's schedule and optimal schedule during the
loops. Hence computing the optimal non-adaptive strategy of the form $T^*E^*$
can be done in time $O(n^3)$. We show how this complexity can be reduced.

\begin{theorem}
	For the non-adaptive model, the equilibrium schedule (and therefore the optimal algorithm's strategy) can be computed in 
	time $O(n^2)$.
\end{theorem}
\begin{proof}
	We observed earlier that all equilibrium schedules consist of four parts and are of the form
	$(Tx)^{d}(Tp)^{a-b} (Ex)^{f-b} (Ep)^{n-a+d-f}$ for some parameters $a,d,f$.  The parameter
	$a$ describes the number of tests done by the algorithm, while $c$ describes
	the total number of long jobs, and $d$ the number of tested long jobs, as
	decided by the adversary.  We fix the parameters $a,f$ and show that the
	optimal parameter $d$ for the adversary can be computed in constant time. As
	there are only $O(n^2)$ possible values for $a, f$ this would prove the
	theorem.

	The parameter $f$ defines the total number of long jobs, hence the optimal schedule, which is $(Ep)^{n-f}(Ex)^{f}$, does not depend on $d$.  As a result, the adversary chooses $d$ such that the cost of the algorithm is maximum.
	The cost of the algorithm can be expressed as follows
	\begin{align*}
		\textrm{ALG(d)} =&\phantom{+} 
			d (d + 1)/2 \tag{cost of part 1}			\\ &
		  +   d n \tag{delay caused by part 1}		  \\ &
		  + (1 + p) (a - d) (a - d + 1)/2 \tag{cost of part 2} \\ &
		  + (1 + p) (a - d) (n - a + d) \tag{delay caused by part 2}\\ &
		  + (p + x) (f - d) (f - d + 1)/2 \tag{cost of part 3}\\ &
		  + (p + x) (f - d) (n - a + 2 d - f)\tag{delay caused by part 3}\\ &
		  + p (n - a + d - f) (n - a + d - f + 1)/2 \tag{cost of part 4}\\ &
		  +   p (n - a - f + d) d \tag{delay caused by part 4}\\ &
		  + (p + x) d (d + 1)/2. \tag{cost of delayed long jobs}
	\end{align*}
	This function is quadratic in $d$ with a negative second derivative, namely
	$-2x$. Therefore the integer maximizer of ALG can be found, by first finding
	the root of the function $d \mapsto ALG(d) - ALG(d-1)$ and  then rounding it
	down.  This function is $a + (a - 2 d + 2 f - n + 1)  x$, and its root is
	\[
		d^* = f -\frac{n-a-a/x -1}{2},
	\]
	which never exceeds $f$. The adversary has to choose $d$, such that $0\leq
	d\leq \min\{a,f\}$, hence he chooses
	\[
		\min\{a, \max\{0, \lfloor d^* \rfloor\}\}.
	\]
	Since this integer can be computed in constant time, this concludes the proof.
	The optimal algorithm is summarized in Figure~\ref{fig:algo-opt-non-adaptive}.
\end{proof}

\begin{algorithm}
\begin{algorithmic}[0]
	\State{best $= +\infty$ } 
	\State $a^* =$ None  
	\ForAll{$0\leq a, f\leq n$}
		\State{OPT$ = pn(n+1)/2 + x f(f+1)/2$}
		\State{$d^* = f - (n-a-a/x-1)/2$}
		\State{$d = \min\{a, \max\{0, \lfloor d^* \rfloor\}\}$}
		\State{Ratio = ALG(d)/OPT }
		\If{Ratio $<$ BestRatio} \Comment{keep best ratio}
		\State{BestRatio = Ratio}
		\State{$a^* = a$ }
		\EndIf
	\EndFor
	\State \Return{``The optimal strategy tests the first $a^*$ jobs.''} 
\end{algorithmic}

\caption{Pseudocode computing the optimal strategy for the algorithm in the non-adaptive model.}
\label{fig:algo-opt-non-adaptive}
\end{algorithm}

We conclude this section, by providing a closed expression for the asymptotic
competitive ratio. The computations have been conducted with the help of the
computing system Mathematica. The files are available to the reader in the
companion webpage
\url{https://www.lip6.fr/Christoph.Durr/Scheduling-w-Oracle/}.

Formally we analyze the limit of the competitive ratio when
$n$ tends to infinity. This is possible, because for large values of $n$, we
can focus on the dominant parts in the schedule costs, and ignore integrality
of the parameters.  In the sequel instead of working with integral parameters
$a,d,f$ we work with fractions.  Note that along the way we also made a
variable change, in the sense $a=\alpha n, d=\beta n, f=(\beta + \gamma)n$.

The intuition behind the following theorem is that for large instances, if $x$ is not
larger than $2+1/p$ then the difference of short and large jobs is so small
that the length of a test is too costly compared to the gain obtained from the
returned information.  

\begin{theorem}
  The asymptotic competitive ratio of the scheduling problem for non-adaptive algorithms is
  \[ \sqrt{1+\frac{x}{p}} \]
  when  $0\leq x < 2 + 1 / p$ and
  \[ 1 + \frac{x^2 - p x - 1 + \sqrt{\Delta'}}{2 p
  x^2} \]
  when $x \geq 2 + 1 / p$,
  for $\Delta'=8 p (x - 1) x^2 + (1 + p x - x^2)^2$.
\end{theorem}
\begin{proof}
Given $n$, the number of jobs, the algorithm decides on some fraction $0
\leq \alpha \leq 1$, and will test the first $\alpha n$ jobs,
followed by the untested execution of the remaining $(1 - \alpha) n$ jobs. The
adversary decides on some fraction $0 \leq \beta \leq \alpha$ of the
tested jobs to be long, and some fraction $0 \leq \gamma \leq 1 -
\alpha$ of the executed jobs to be long. Both the cost of the algorithm and of
the optimal schedule are quadratic expressions in $n$. For the sake of
simplicity we focus only on the quadratic part, which is dominating for large
$n$.  Hence the results of this section apply only asymptotically when $n$
tends to infinity.  However in principle it is possible to make a similar but
tedious analysis also for the linear part and obtain results that hold for all
job numbers $n$.

\begin{figure}[!ht]
  \centerline{\includegraphics[width=14cm]{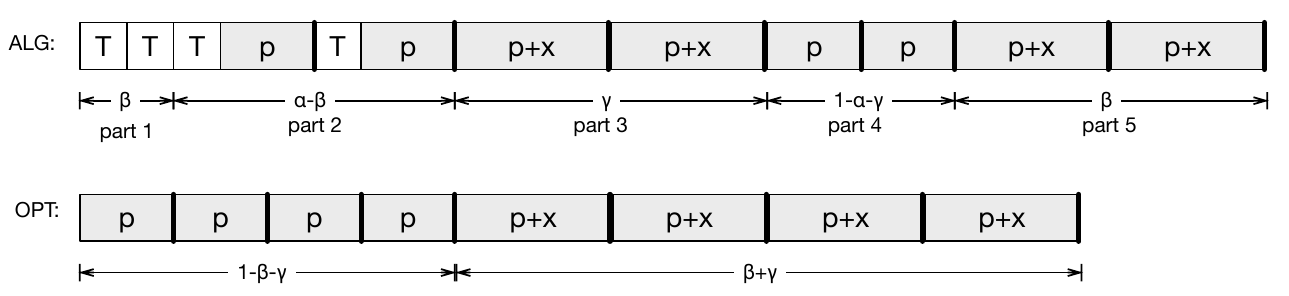}}
  \caption{Schedule obtained in the non-adaptive setting.\label{fig:TE_abc}}
\end{figure}

The game played between the algorithm and the adversary contains the following strategies. The
algorithm chooses some parameter $\alpha$ while the adversary chooses some
parameters $\beta, \gamma$. The $n^2$ dependent part of the cost of the
algorithm's schedule is (we multiply by 2 to simplify notation)
\begin{eqnarray}
  2 \cdot \textrm{ALG} & = & \phantom{+} 2 \beta  \label{a1}\\  
  &  & + (1 + p) (\alpha - \beta)^2  \label{a2}\\
  &  & + 2 (1 + p) (\alpha - \beta) (1 - \alpha + \beta)  \label{a3}\\
  &  & + (p + x) \gamma^2  \label{a4}\\
  &  & + 2 (p + x) \gamma (1 - \alpha + \beta - \gamma)  \label{a5}\\
  &  & + p (1 - \alpha - \gamma)^2  \label{a6}\\
  &  & + 2 p (1 - \alpha - \gamma) \beta  \label{a7}\\
  &  & + (p + x) \beta^2 .  \label{a8}
\end{eqnarray}
We justify the cost, by distinguishing the contribution of each of the parts
of the schedule to its cost, referring to the 5 parts of the schedule
illustrated in Figure \ref{fig:TE_abc}. Expression {\eqref{a1}} is the delay
caused by part 1 on the following parts of the schedule, {\eqref{a2}} is the
cost of part 2, {\eqref{a3}} is the delay caused by part 2 on the rest,
{\eqref{a4}} is the cost of part 3, {\eqref{a5}} is the delay of part 3 on the
rest, {\eqref{a6}} is the cost of part 4, {\eqref{a7}} is the delay of part 4
on part 5, {\eqref{a8}} is the cost of part 5.

The $n^2$ dependent part of the optimal schedule (again multiplying by $2$ to
simplify notations) is
\[
  2\cdot\textrm{OPT}  =  \phantom{+} p  + x(\beta + \gamma)^2.
\]
The algorithm wants to minimize the ratio $\textrm{ALG} / \textrm{OPT}$ while the
adversary wants to maximize it. In order to avoid the manipulation of
fractions, we introduce another parameter $r$ chosen by the algorithm. It
defines the expression
\begin{align*}
 G &= (1 + r) 2 \cdot \textrm{OPT} - 2 \cdot \textrm{ALG}
\\	
 &= p r + \alpha^2 + \beta^2 - 2 \alpha (1 + \beta) + 2 x \gamma (\alpha +
   \gamma - 1) + r x (\beta + \gamma)^2 . 
\end{align*}
and in this setting, $G$ is negative if and only if the ratio is strictly
larger than $1 + r$.   Now we have reformulated the game. The algorithm wants
to have $G$ equal to zero and can choose $\alpha$ and $r$ for this purpose. 
In fact we can think of the algorithm wanting to maximize $G$, and choosing
the smallest $r$ such this goal is achieved for $G$ being $0$. In contrast the
adversary wants to have $G$ negative, and we can think of the adversary
wanting to minimize $G$.  For this purpose we can choose $\beta, \gamma$.

We analyze this game by fixing $\alpha$ and $r$ and understanding the best
response of the adversary. Since $G$ is a quadratic expression in $\beta$ and
$\gamma$ this best response can be either the minimizer of $G$, i.e.\ the
value at which sets the derivative of $G$ to zero, or one of the boundaries of
the domains of $\beta$ and $\gamma$.  Once we fixed the best response of the
adversary we can analyze the best response of the algorithm to the choice of
the adversary.  This leads to a simple but tedious case analysis which we
present now.  We start by breaking the analysis into two cases depending on
how $\beta$ compares to $2+1/p$, because this decides if the extreme point of
$G$ in $\beta$ is a local minimum or maximum.

\subsection{Case $x \geq 2 + 1 / p$}

For fixed $\alpha, r$ what would be the best response by the adversary? Note
that $G$ is convex both in $\beta$ and $\gamma$. Hence one possibility for the
adversary would be to choose the extreme points, provided that they satisfy
the required bounds $0 \leq \beta^* \leq \alpha$ and $0
\leq \gamma^* \leq 1 - \alpha$. The extreme point for $\beta$
is
\begin{eqnarray*}
  \beta^* & = & \frac{\alpha - r x \gamma}{1 + r x} .
\end{eqnarray*}
Choosing $\beta = \beta^*$ we observe that $G$ remains convex in
$\gamma$, since the second derivative of $G$ in $\gamma$ is
\[ 2 x \left( 2 + \frac{r}{1 + r x} \right) . \]

Hence we consider the extreme point for $\gamma$ which is
\begin{eqnarray*}
  \gamma^* & = & \frac{(1 + r x) (1 - \alpha) - r \alpha}{2 + r + 2 r x}
  .
\end{eqnarray*}
Choosing $\gamma = \gamma^*$ the expression $G$ writes as

\[ G = pr+\frac{- 2 \alpha (r + 2) - (1 - \alpha)^2 r x^2 + x
   (\alpha (\alpha + 2 r - 2) + 1)}{2 + r + 2 r x} 
   . \]
We observe that $G$ is concave in $\alpha$ hence one possibility for the
algorithm is to choose the extreme point for $\alpha$, which is
\[ \alpha^* = \frac{x + r x^2 - r x - r - 2}{x + r x^2} . \]
Choosing $\alpha = \alpha^*$ the expression $G$ writes as
\begin{align*}
 G &= \frac{2 + r + (p r - 2) x + r (p r - 1) x^2}{x + r x^2} \\
 &= \frac{2-2x + r(1+px-x^2) + r^2 px^2}{x + rx^2}.
\end{align*}
The algorithm chooses $r$ such that $G$ becomes zero. For this purpose we
consider the roots of the numerator of $G$, which are
\begin{eqnarray*}
  r_1^* & = & \frac{x^2 - p x - 1 - \sqrt{\Delta'}}{2 p x^2}\\
  r_2^* & = & \frac{x^2 - p x - 1 + \sqrt{\Delta'}}{2 p x^2},
\end{eqnarray*}
for
\begin{equation}
	\Delta'=8 p (x - 1) x^2 + (1 + p x - x^2)^2. 
  \label{eq:delta_prime}
\end{equation}

By case assumption $x\geq 2 + 1/p$, both roots of $G$ are real.
Since $G \geq 0$ means that the algorithm has ratio at most $1 + r$, and since
the numerator of $G$ is concave in $r$ the algorithm has to choose the larger
of both roots. Both roots are real, since $x\geq 1$ implies $\Delta'\geq 0$.
We choose the root $r^*_2$ which defines the ratio of the game provided the
following conditions are satisfied:
\[ r \geq 0, \: 0 \leq \beta^* \leq a^*, \: 0
   \leq \gamma^* \leq 1 - \alpha^* . \]

\paragraph{Condition $r \geq 0$}

The expression $r^*_2$ has a single root in $x$ namely at $x = 1$.
Moreover its limit when $x$ tends to infinity is $1 / p$. Hence $r \geq
0$ holds by case assumption $x \geq 2 + 1 / p > 1$.

\paragraph{Condition $\beta \geq 0$}

With the choices for $\alpha, \gamma$ and $r$ the value $\beta^*$ writes
as
\[ \beta^* = \frac{- x^2 - p x - 1 + \sqrt{\Delta'}}{2 x^2} = r_2^* p - 1. \]
This means that the condition $\beta \geq 0$ translates into $r_2^*
\geq 1 / p$. We observe that $r_2^* = 1 / p$ has a single solution
in $x$, namely $x = 2 + 1 / p$. Since at $x = 1$ the value of $r_2^*$ is
smaller than $1 / p$, we know that for $x \geq 2 + 1 / p$ we have
$r^*_2 \geq 1 / p$ and hence $\beta \geq 0$.

\paragraph{Condition $\beta \leq \alpha$}

We study the difference which is
\[ \alpha^* - \beta^* = \frac{1 + (- 4 + p) x + 3 x^2 - \sqrt{\Delta'}}{x^2 - 2 x} . \]
There is no risk of dividing by $0$ in the range $x \geq 2 + 1 / p$. We
observe that the numerator has two roots in $x$, namely $x = 0$ and $x = 1$.
Moreover the numerator evaluates as $(2 + 4 p) / p^2$ at $x = 2 + 1 / p$. We
conclude that $\beta^* \leq \alpha^*$ for the range $x
\geq 2 + 1 / p$.

\paragraph{Condition $0 \leq \gamma$}

The expression $\gamma$ writes as
\[ \frac{- (x - 1)^2 - p x + \sqrt{\Delta'}}{x^3
   - 2 x^2}, \]
again there is no risk of dividing by $0$ in the range $x \geq 2 + 1 /
p$. The numerator evaluates to zero at $x = 0, x = 1$ and at $x = 1 / (1 +
p)$. Since $\gamma^*$ has the value $p^2 / (1 + 3 p + 2 p^2) > 0$ at $x =
2 + 1 / p$ we obtain that $\gamma^* \geq 0$ in the range $x
\geq 2 + 1 / p$.

\paragraph{Condition $\gamma \leq 1 - \alpha$}

We observe that the expression $1 - \alpha - \gamma$ simplifies as $1 / x$.

This concludes the verification of the conditions on $r, \alpha, \beta,
\gamma$ and shows that for $x \geq 2 + 1 / p$ the competitive ratio is
$r^*_2$.

\subsection{Case $x < 2 + 1 / p$}

We have observed earlier that the extreme point $\beta^*$ is negative
when $x < 2 + 1 / p$. Hence for this range of $x$, the adversary chooses
$\beta = 0$. This means that all tested jobs will be short and the algorithm
has no incentive to test jobs, and will just execute them untested. We justify
this intuition by a formal analysis. For the choice $\beta = 0$, the
expression $G$ reads as
\[ G = p r - 2 \alpha + \alpha^2 + x \gamma (r \gamma - 2 (1 - \alpha -
   \gamma)), \]
which is convex in $\gamma$. Hence the adversary will choose the extreme point
in $\gamma$ which is
\[ \gamma^* = \frac{1 - \alpha}{2 + r} . \]
Note that this value satisfies $0 \leq \gamma^* \leq 1 -
\alpha$ as required. With this choice of $\gamma$ the goal function becomes
\[
	pr(2+r) -x(1- \alpha)^2 - (2+r)(2- \alpha)\alpha
\]
where we multiplied the goal function by $2+r$, preserving its sign and
simplifying the expression.  The second derivative in $\alpha$ is
\[
	4+2r-2x,
\]
which could be positive or negative. 

We start with the non positive case, meaning $2+r\leq x$.  In this case the algorithm chooses the extreme value, which is $\alpha=1$. In other words the algorithm tests all jobs, which all happen to be short and the competitive becomes $(1+p)/p=1+1/p$.  But $2+r=3+1/p$, which contradicts the case assumptions $2+r\leq x<2+1/p$. 

Finally we continue with the positive case.
Here the algorithm chooses the lower bound for
$\alpha$, which is $0$, translating the above mentioned intuition that the
algorithm will execute all jobs untested.

For $\alpha = 0$, the expression $G$ simplifies as $p r - x / (2 + r)$ which
is set to $0$ by $$r = \sqrt{(p + x) / p} - 1.$$

This concludes the case analysis, and therefore the proof of the theorem.
\end{proof}

\section{Adaptive model}
\label{sec:adaptive}

In this section we analyze two-phase algorithms in the adaptive model, as
described in the introduction. This means that the algorithm first tests some
jobs, then executes untested the remaining jobs.

It is convenient to illustrate the interaction between such an algorithm and the adversary
by a walk on a grid as follows. The vertices of the grid consist of all points
with coordinates $(c,d)$ such that $0\leq c, d$ and $c+d\leq n$.  Cells $(c,d)$
with $c+d=n$ are called \emph{final} cells, and non final cells $(c,d)$ are
connected to the cells $(c+1,d)$ and $(c,d+1)$. These arcs form a directed
acyclic graph with root $(0,0)$, which is the upper left corner of the grid,
see Figure~\ref{fig:grid}.

The walk starts at the root, and follows only down or right steps to adjacent
cells.  If the algorithm decides to test a job, the adversary can choose to
respond with a long job, resulting in a right step, or with a short job,
resulting in a down step.  In this sense the adversarial strategy played
against a fictive algorithm testing all jobs, translates into a path $P$
connecting the root to some final cell.  

If at some moment the current position is $(c,d)$, then it means that the algorithm
tested $c+d$ jobs, among which $d$ were long and $c$ were short.  These two
integers are not enough to fully describe the cost of the schedule produced by
the algorithm. Therefore we introduce an additional value, called the 
\emph{test cost}. Given a path $P$ from the root to a cell $(c,d)$, the test
cost associated to this cell is defined as  the value $e=\sum_{(c',d')} (n-c')$,
where the sum is taken over all grid cells $(c',d')$ visited by the walk $P$
between the root and the cell $(c,d)$ but excluding the final cell $(c,d)$.
This test cost $e$ represents the increase in cost due to the different tests
done so far,  since each test delays by one all $n-c'$ subsequent job
executions, i.e.\ its rank is $n-c'$.   We associate to the cell $(c,d)$ of
the path $P$ a so called \emph{stop ratio}   
\[   
R(c,d,e) =
\max_{0\leq b < n-c-d} \frac{\textrm{ALG}(b,c,d,e)}{\textrm{OPT}(b,d)},   
\]
for  
\begin{multline*}  \textrm{ALG}(b,c,d,e)   = p\frac{n(n+1)}2+ e \\
+ x\frac{   (n-c)(n-c+1) - (n-c-b)(n-c-b+1) + d(d+1)}2 
\end{multline*}  
\[
\textrm{OPT}(b,d)  =  p\frac{n(n+1)}2 + x\frac{(b+d)(b+d+1)}2.  
\]

Note that the stop ratio $R(c,d,e)$ is associated to the cell $(c,d)$ but
depends also on the fixed path $P$, since it determines the value $e$.

This is the ratio reached by the game if the algorithm decides to switch to
the execution phase after $c+d$ tests.  Here $b$ represents the number of long
executed jobs (i.e. $Ex$) in the execution phase and ALG, OPT are the costs of
the respective schedules, see Figure~\ref{fig:algo_n2_cost} for illustration.

\begin{figure}[htb]
  \centerline{\includegraphics[width=10cm]{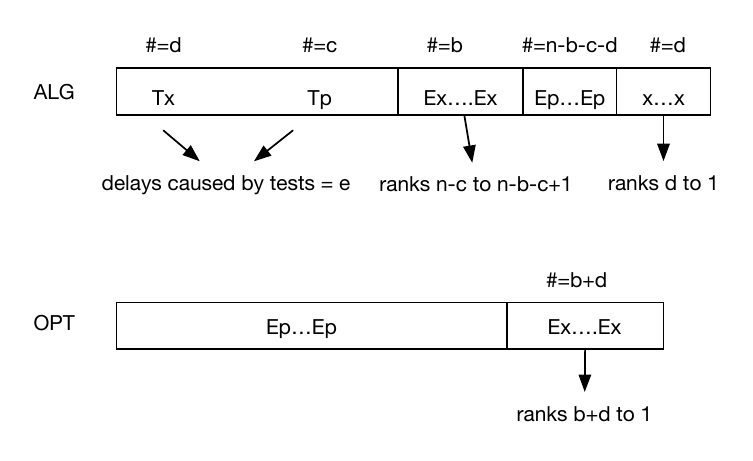}}
  \caption{Illustration of the cost expression.}
  \label{fig:algo_n2_cost}
\end{figure}

We explain now, how $R(c,d,e)$ can be computed efficiently.
\begin{lemma} 	\label{lem:R}
	The stop ratio $R(c,d,e)$ can be computed in constant time.
\end{lemma}
\begin{proof}
The fraction $\textrm{ALG}(b,c,d,e) / \textrm{OPT}(b,d)$ is maximized over $b$, since the adversary gets to choose $b$
and wants to maximize the ratio.
 
Both costs ALG$(b,c,d,e)$ and OPT$(b,d)$ are quadratic in $b$. Hence the ratio
ALG/OPT as a function of $b$ is continuous and derivable.
We consider the function 
\[
    g(b)=ALG(b,c,d,e)/OPT(b,d)
\]
and try to identify a maximizer in $\{0,1,\ldots,n-c-d\}$. For this purpose we study the sign of $g(b)-g(b-1)$. Since OPT is positive it is equivalent to analyze the sign of the function 
\begin{align*}
    h(b) =& \frac2x  (ALG(b,c,d,e)OPT(b-1,d)-ALG(b-1,c,d,e)OPT(b,d))\\
         =& -2e(b+d) + p n(n+1)(n-c-d-2b+1) 
         \\&
         - x(b+d)(b(n-c+d+1) - (1+d)(n-c-d+1))
\end{align*}
where the purpose of the factor $2/x$ is to simplify the expression.  
Its second derivative in $b$ is $-2x(n-c+d-1)<0$. Hence $h$ is quadratic in $b$, and first increasing, then decreasing. This means that $g(b) - g(b-1)$ is also first increasing in $b$, then decreasing. We distinguish different cases.

In case $h$ is non-positive, we know that $g$ is non-increasing, and hence maximized at $b=0$.
Otherwise $h$ has two roots $b_1 < b_2$.  At this point we consider $\lfloor b_2 \rfloor$.  If it is negative or less than $b_1$, then we know that $h$ is non-positive at integral values $b \geq 0$, and again $g$ is maximized at $b=0$.

In the remaining case, $\lfloor b_2 \rfloor$ is non-negative and at least $b_1$. Hence $\lfloor b_2 \rfloor$ is the last integer where $h$ is non-negative and therefore is the integral maximizer of $g$. The roots are
\[
	b_{1,2}=\frac{(n-2d^2-c-d+1)x-2pn(n+1)-2e \pm \sqrt{\Delta}}{2x(n-c+d+1)}
\]
for
\begin{align*}
	\Delta= & ((n-2d^2-c-d+1)x-2pn(n+1)-2e)^2 \\
        & +4x(n-c+d+1)((n-c-d+1)(pn(n+1)+xd(d+1))-2de).
\end{align*}
We observe that 
\[
    h(n-c-d) = -2e(n-c)-pn(n+1)(n-c-d-1)-x(n-c)(n-c+1)(n-c-d-1)
\]
is negative, since $c+d<n$ is implied by Lemma~\ref{lem:suffix_Ep}. This implies $\lfloor b_2 \rfloor \leq n-c-d$, as required.

In summary, if $\Delta\leq 0$ or  $\lfloor b_2 \rfloor < \max\{0,b_1\}$, then $g$ is maximized at $b=0$, otherwise it is maximized at $b=\lfloor b_2 \rfloor$.  As a result the stop ratio $R(c,d,e)$ can be computed in constant time.
\end{proof}


\begin{figure}[h!]
	\begin{center}
	\begin{tikzpicture}[scale=0.5]
\foreach \x/\y/\c in {0/0/white,1/0/white,2/0/black,3/0/gray!30,4/0/gray!30,5/0/gray!30,6/0/gray!30,7/0/gray!30,0/1/white,1/1/white,2/1/white,3/1/gray!30,4/1/black,5/1/black,6/1/gray!30,0/2/white,1/2/white,2/2/white,3/2/black,4/2/gray!30,5/2/gray!30,0/3/white,1/3/white,2/3/white,3/3/white,4/3/white,0/4/white,1/4/white,2/4/white,3/4/white,0/5/white,1/5/white,2/5/white,0/6/white,1/6/white,0/7/white}           \draw[fill=\c]  (\y,-\x) rectangle ++(1,1); 
\draw[|->,color=red] (0.5,0.5) -- (0.5,-0.5) -- (1.5,-0.5) -- (1.5,-1.5) -- (3.5,-1.5) -- (3.5,-3.5);
\draw (4,1.6) node {$d$};
\draw (-0.5,-3) node {$c$};
\end{tikzpicture}
	\end{center}
	\caption{The grid as used in the procedure to compute the equilibrium schedule. 
	Marked cells are black or gray. In red the boundary path.}	
	\label{fig:grid}
\end{figure}
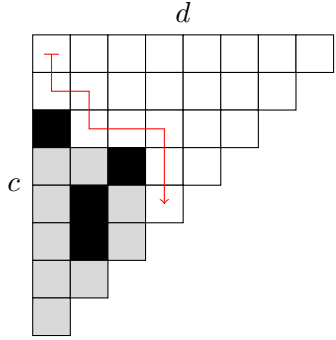

Finally we explain how to determine the optimal algorithmic strategy.
\begin{theorem}
	Assuming Conjecture~\ref{conjecture}, the optimal strategy for the algorithm in the adaptive setting can be computed in time $O(n^3)$.
\end{theorem}
\begin{proof}
	We first show how to compute an optimal adversarial strategy, in form of a
path $P$ from the root to some final cell. Denote by $R^*$ the minimal stop
ratio along $P$. The algorithm will then walk along $P$, testing a new job at
each step, until it reaches the cell with stop ratio $R^*$, at which point it
switches to the execution phase.  In case the adversary behavior differs from
$P$, the procedure needs to restart and computes an optimal strategy for the
adversary based on the current situation. For the ease of
presentation we don't describe this generalized procedure.

	Our algorithm to compute the optimal adversarial strategy works as follows.
	At any moment in time it stores a path $P$ and its minimal stop ratio $R^*$
	and repeatedly tries to find a path with larger minimal stop ratio.  In the
	formal description of Algorithm~\ref{alg:adaptive} we choose to represent $P$
	as a list of tuples $(c,d,e)$, where $(c,d)$ is a grid cell and $e$ the
	associated test cost.  

	Some cells of the grid are marked with the meaning that any path traversing
	any of the marked cells is guaranteed to have a minimal stop ratio at most
	$R^*$.  The marked cells form a \emph{combinatorial tableau} in the sense,
	that if cell $(c,d)$ is marked, then so are all cells $(c',d')$ with $c'\geq
	c$ and $d'\leq d$.

	The procedure stops when the root cell $(0,0)$ is marked and returns $P$ as
	the optimal adversarial strategy.

	For now assume that the root cell is not marked. Then there is a path $P'$
	which follows the boundary of the marked cells.  This path has the property
	that for every cell $(c,d)$ on that path, the associated test cost $e$ is
	maximal among all paths that don't traverse a marked cell and therefore
	maximizes the stop ratio at $(c,d)$.  Indeed, if one replaces two steps
	$(c',d')\rightarrow(c',d'+1)\rightarrow(c'+1,d'+1)$ in a path by
	$(c',d')\rightarrow(c'+1,d')\rightarrow(c'+1,d'+1)$ then the test cost of the
	cell $(c'+1,d'+1)$ is increased by $1$ and so are the test costs of the
	subsequent cells on the path. The only path which does not allow this
	modification is the boundary path $P'$.

	Two things can happen with $P'$. Either it contains a  cell  with stop ratio at
	most $R^*$.  In what case this cell can be marked.  We can safely mark as
	well as all cells to its left and below, because any boundary path traversing
	these cells must traverse $(c,d)$ as well, even if more cells get marked
	later on. 

	Or the path has minimal stop ratio $R'$ strictly larger than $R^*$. In what
	case $P'$ is selected as the current best known adversarial strategy and
	$R^*$ is increased to $R'$.  As a result the cell on $P'$ with stop ratio
	$R^*$ will be marked in the next iteration as well. Hence a measure of
	progress of this procedure is the number of marked cells, which limits the
	number of iterations to $O(n^2)$, and probably much less in practice.

	Inspecting the boundary path takes $O(n)$ time. The total time spend on
	marking cells is $O(n^2)$. Indeed when marking cells to the left and below a
	cell $(c,d)$ on the boundary path, we know that cell $(c+1,d-1)$ is already
	marked, and it suffices to mark all cells $(c,d-1),(c,d-2),\ldots$, until the
	boundary of the grid or a marked cell is reached and to  to mark all cells
	$(c+1,d),(c+2,d),\ldots$, until the boundary of the grid or a marked cell is
	reached.  As a result, the time spend on marking cells is constant for each cell.
	This means that the overall time complexity is $O(n^3)$.
\end{proof}

\begin{algorithm}
\caption{Algorithm computing the optimal strategy in the adaptive model}
\label{alg:adaptive}
	\begin{algorithmic}[0]
	\Procedure{BoundaryPath}{$R^*,\textrm{grid}$}
	\State $(c,d,e)=(0,0,0)$
	\State Initially $P$ contains only $(0,0,0)$
	\While{$c + d < n$}					
		\If{$(c+1,d)$ is not marked}	\Comment{privilege down steps}
			\State $c=c+1$
		\Else
			\State $d=d+1$
		\EndIf
		\If{$R(c,d) \leq R^*$}	\Comment{boundary path does not improve $R^*$}
			\State mark all cells $(c',d')$ with $c'\geq c, d'\leq d$
			\State \Return None
		\EndIf
		\State append $(c,d,e)$ to $P$
		\State $e=e+c$
	\EndWhile
	\State \Return{$P$}
	\EndProcedure	

	\Procedure{OptimalAdversary}{}
	\State grid = $n\times n$ grid without marked cells
	\State $R^* = 1$
	\While{cell (0,0) is not marked}
		\State $P$ = BoundaryPath($R^*,\textrm{grid}$)
		\If{$P$ is not None}					\Comment{$P$ improves $R^*$}
			\State $R^*$ = minimum stop ratio along $P$
			\State best = $(P,R^*)$		\Comment{best path-ratio pair}
		\EndIf
	\EndWhile
	\State \Return best
	\EndProcedure
	
	\Procedure{OptimalAlgorithm}{}
	\State $(P,R^*) = \textrm{OptimalAdversary()}$
	\State $(c,d,e) = (0,0,0)$
	\While{$R(c,d,e) > R^*$}			
		\State Test next job, execute immediately if it is short
		\State $(c, d, e)$ is next cell on $P$
		\Comment assumes adversary follows optimal strategy $P$.
	\EndWhile			\Comment{stop at minimum stop ratio}
	\State Execute untested all remaining jobs, followed by the $d$ postponed tested long jobs
	\EndProcedure	
	\end{algorithmic}
\end{algorithm}

\section{Experiments}
\label{sec:experiments}

We conducted experiments to verify Conjecture~\ref{conjecture}, for up to 10
jobs, 128 uniformly spread values\footnote{We used a power of two in the hope
of reducing rounding errors.  In addition we conducted experiments using the
Python module \textsl{fractions} which completely avoids these errors.}
$p\in(0,100]$, as well as 128 uniformly spread values $x\in(0,10]$, and for
both the adaptive and non-adaptive model.  By analyzing the game tree of each
instance, no counterexample to Conjecture~\ref{conjecture} was found.  The
programs used for the experiments are provided in the companion webpage
\url{https://www.lip6.fr/Christoph.Durr/Scheduling-w-Oracle/}. Note that these
experiments explore the game tree which has $2^{2n}$ leafs, so $n=10$ is
roughly the limit up to which we can test the conjecture.  The algorithm from
the previous section is based on the two-phase assumption and has theoretical
complexity $O(n^3)$, which in practice is $O(n^2)$ because it seems to
terminate after a constant number of iterations.  Hence an implementation of
this algorithm can clearly handle roughly $10.000$ jobs.

Just for curiosity we plot the algorithmic strategies in terms of number of
tested jobs, see Figures~\ref{fig:exp-n6} and \ref{fig:exp-n6:more}. It is
interesting to observe that the strategies are quite different in the adaptive
and in the non-adaptive model.  We also plot the competitive ratio, and not
surprisingly observe that it is worse when $x$ is large and $p$ small, while
tending to $1$ for larger values of $p$.  Another measure that is interesting
to extract from these experiments, is the gain of adaptivity. This is defined
as the competitive ratio in the non-adaptive model compared to its counterpart
in the adaptive model, and has been studied theoretically in contexts of query
algorithms \cite{gupta_adaptivity_2017,dean_adaptivity_2005}.   We can observe
that the problem has a small gain of adaptivity, in the order of 2\%. 

We know that assuming Conjecture~\ref{conjecture}, in the non-adaptive model,
all equilibrium schedules are of the form $(Tx)^*(Tp)^*(Ex)^*(Ep)^*$. We
observe that in the adaptive model for most instances the equilibrium
schedules is also of this form.  Figure~\ref{fig:exp-n6} shows instances where
this is not the case.  For larger value of $n$, we observed quite a variety of
these \emph{noticeable schedules}, as we call them.  They do not seem to obey
a particular structure, which could be exploited in a more time efficient
algorithm for the adaptive model.

\begin{figure}[hp]
	\includegraphics[width=\textwidth]{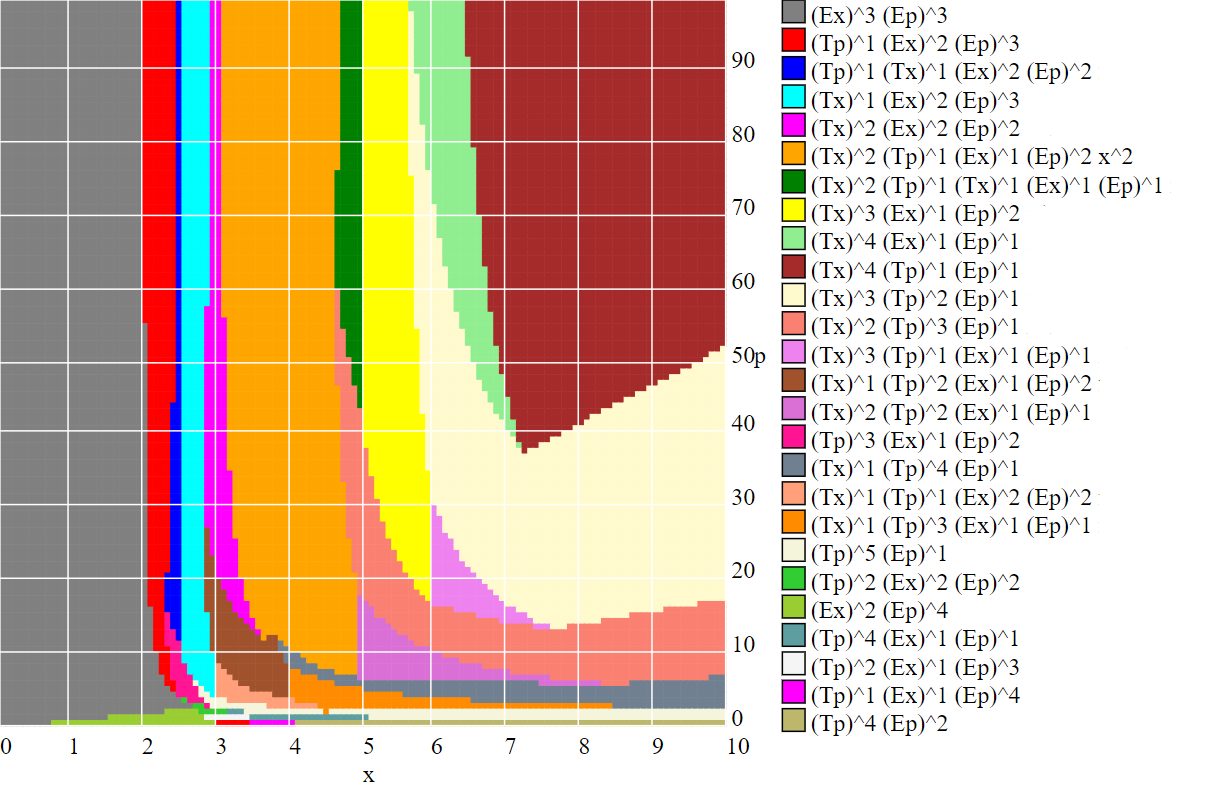}
	\\
	\includegraphics[width=0.7\textwidth]{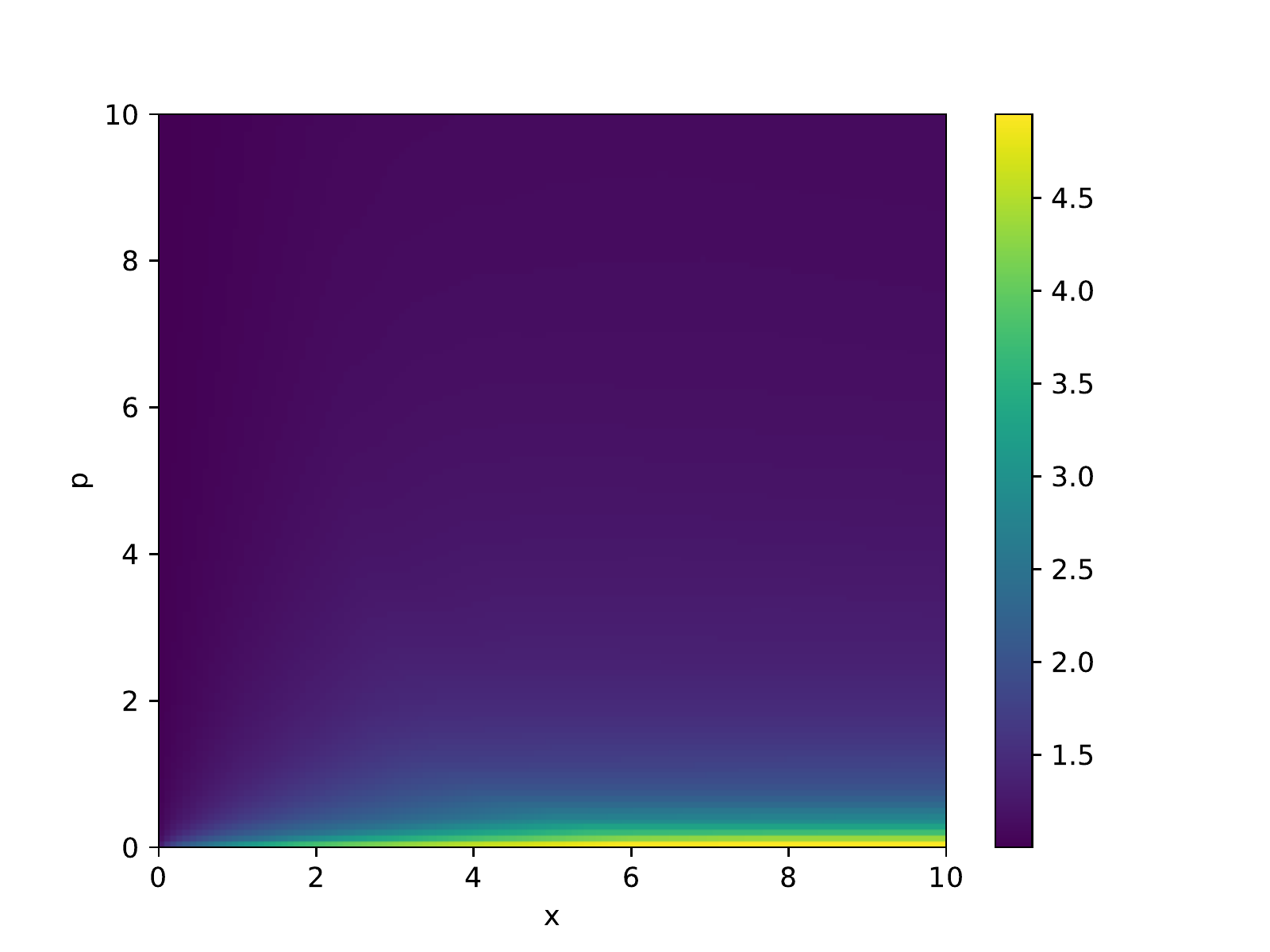}
\caption{Experiments made for $n=6$ jobs. Every $(x,p)$ point corresponds to the game with specific values for $x$ and $p$. For each point the equilibrium schedule is computed, and depicted in the first plot. In this example two equilibrium schedules differ from the pattern $(Tx)^*(Tp)^*(Ex)^*(Ep)^*$, namely $TpTxExExEpEp$, and $TxTxTpTxExEp$.\\
Second plot shows the competitive ratio with the color scheme documented to its right.
}
\label{fig:exp-n6}
\end{figure}

\begin{figure}[hp]
 \begin{minipage}[c]{0.635\textwidth}
	\noindent
	\includegraphics[width=\textwidth]{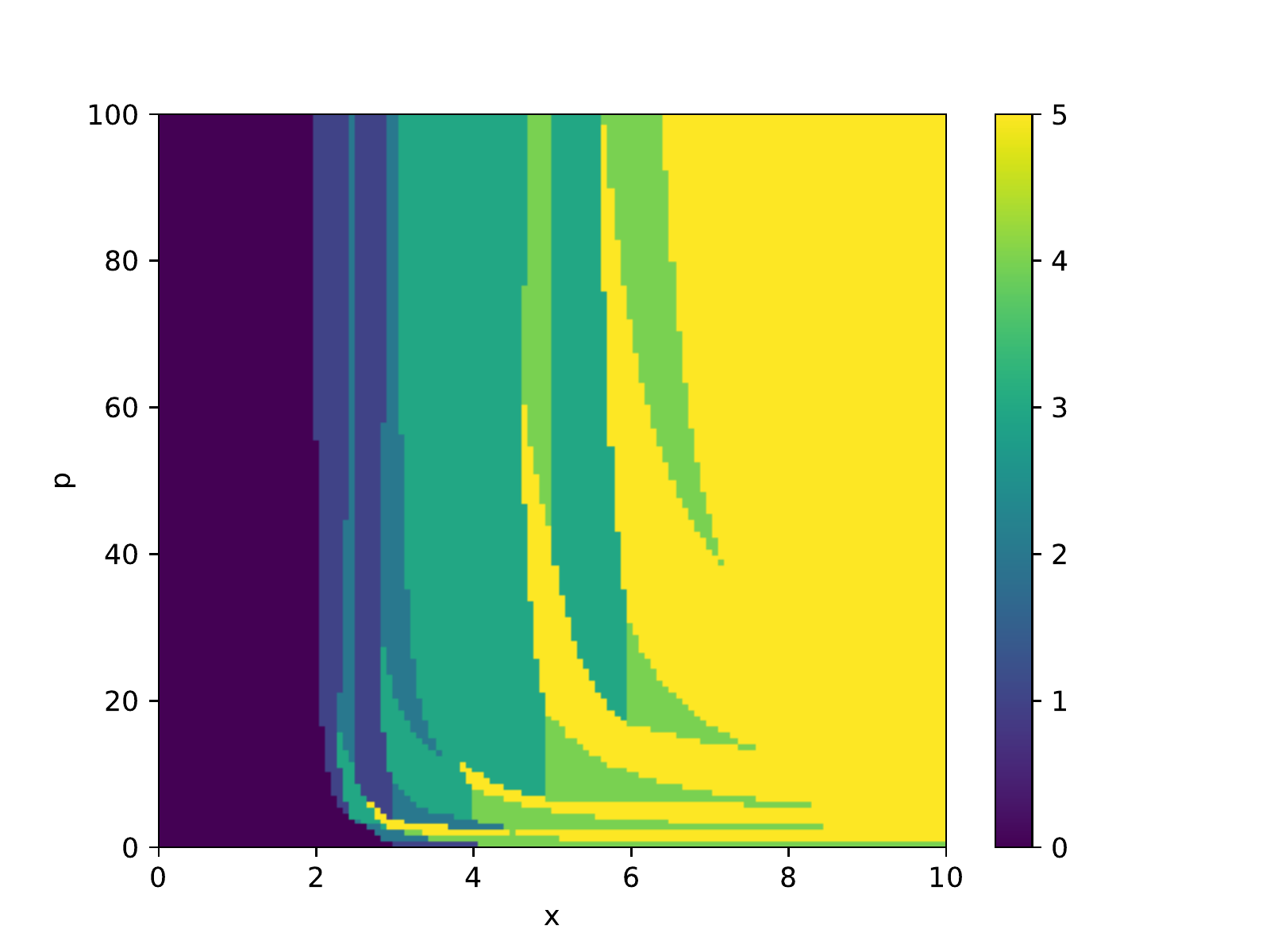}
	\\
	\includegraphics[width=\textwidth]{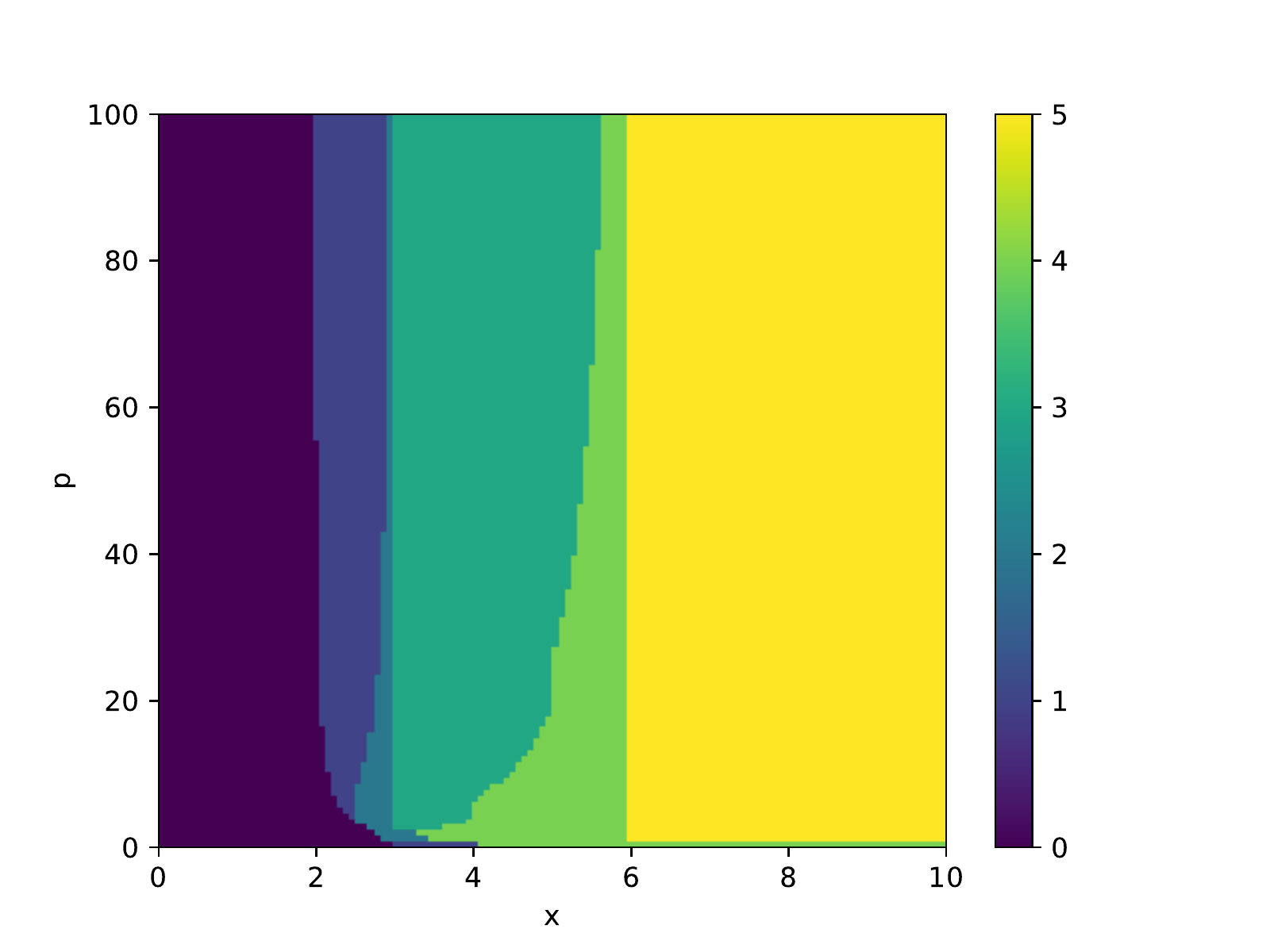}
	\\	
	\includegraphics[width=\textwidth]{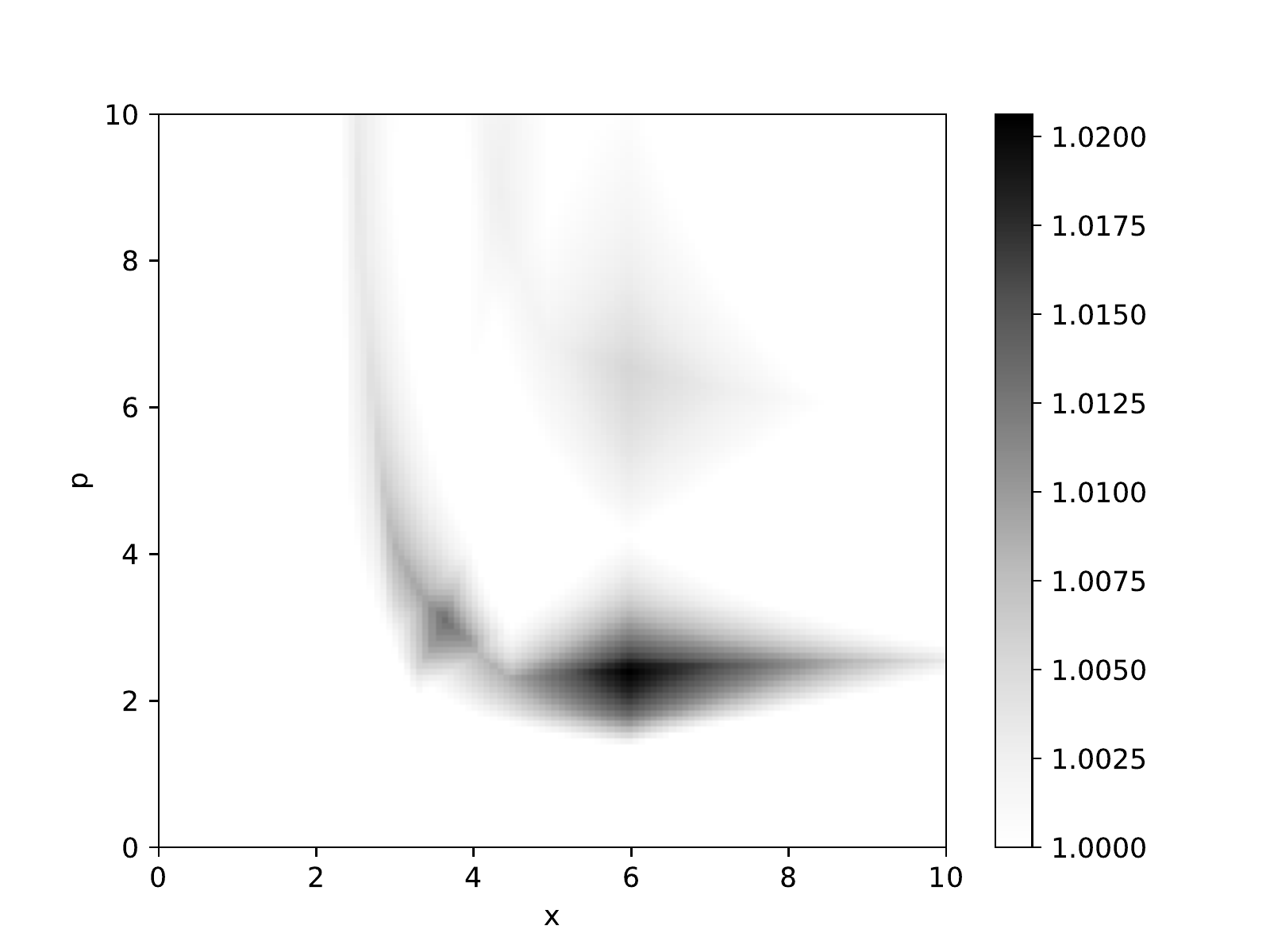}
  \end{minipage}\hfill
  \begin{minipage}[c]{0.3\textwidth}
\caption{More experiments made for $n=6$ jobs. The first plot is a simplification of the first plot of Figure~\ref{fig:exp-n6}, showing only the number of tests done by the algorithm in the equilibrium schedule. 
\\
 The second plot compares the situation with the non-adaptive setting.  It is interesting to observe that the region of points $(x,p)$ for which the equilibrium schedule has the same fixed number of tests seems to be connected in the non-adaptive case, while it is not in the adaptivity case.
\\
The third plot shows the gain of adaptivity, with a different scale for $p$ for improved readability.  It is interesting to see that it is not monotone in $p$ nor in $x$.}
\label{fig:exp-n6:more}
  \end{minipage}
\end{figure}

\clearpage
\section{Conclusion}

We leave Conjecture~\ref{conjecture} as an open problem. Future research
directions include the improvement of the running times, for example by a
subtle use of binary search.  In addition the next step could be the study of
a more general problem, where every job $j$ is known to have a processing time
in the given interval $[\underline{p}_j, \overline{p}_j]$ and has a testing
time $q_j$.  And finally, randomization clearly helps against the oblivious
adversary, already by initially shuffling the given jobs. Hence it would be
interesting to analyze the randomized competitive ratio.

\section*{Acknowledgment}

We would like to thank Thomas Erlebach for helpful discussions as well as
several anonymous referees  for helpful comments on a previous version of this
manuscript.

\bibliographystyle{plain}
\bibliography{all}

\end{document}